\documentclass[11pt,pdftex,letter]{article}
\usepackage[lined,boxed,commentsnumbered]{algorithm2e}

\usepackage{amssymb}
\usepackage{comment}
\usepackage{amsmath}
\usepackage{fullpage}
\ifx\pdftexversion\undefined
\usepackage[dvips]{graphicx}
\else
  \usepackage[pdftex]{graphicx}
\fi

\def\lf{\tiny}

\def\nnll{\refstepcounter{linenumber}\lf\thelinenumber}
\newcounter{linenumber}

\usepackage[usenames,dvipsnames]{color}
\usepackage{subfigure}
\usepackage{listings}
\usepackage{url}
\usepackage{wrapfig}
\usepackage{cite}
\usepackage{framed,color}
\definecolor{shadecolor}{rgb}{0.9,0.9,0.9}
\usepackage[boxed]{algorithm2e}
\newcommand{\CPO}{\textsc{FixTag}}
\newcommand{\DPO}{\textsc{ReuseTag}}
\newcommand{\PS}{\textsc{PS}}

\definecolor{heraldBlue}{rgb}{0.0,0.0,0.8}
\definecolor{heraldRed}{rgb}{0.8,0.0,0.0}
\definecolor{heraldGray}{rgb}{0.4,0.4,0.4}
\definecolor{heraldBlack}{rgb}{0.0,0.0,0.0} %removes comment color
\definecolor{heraldGreen}{rgb}{0.0,0.4,0.0} %removes comment color

\newcommand{\dom}{\textit{dom}}
\newcommand{\pr}{\textit{pr}}

\newcommand{\seq}{\textit{seq}}
\newcommand{\cur}{\textit{cur}}
\newcommand{\Tag}{\textit{tag}}

\newcommand{\smartparagraph}[1]{\noindent{\bf #1}\ }
\newcommand{\eg}{{\it e.g.}}
\newcommand{\ie}{{\it i.e.}}

\newcommand{\etal}{{\it et al.}\xspace}

\def\NOTES{1}
\def\SAVESPACE{1}

\if \SAVESPACE 1
\usepackage{titlesec}
\titlespacing\section{0pt}{7pt}{6pt}

\fi

\if \NOTES 1
\newcommand{\mcnote}[1]{\textcolor{heraldBlue}{\small \bf [MC: #1]}}
\newcommand{\dlnote}[1]{\textcolor{heraldBlue}{\small \bf [DL: #1]}}
\newcommand{\ssnote}[1]{\textcolor{heraldBlue}{\small \bf [SS: #1]}}
\newcommand{\pknote}[1]{\textcolor{heraldBlue}{\small \bf [PK: #1]}}
\else
\newcommand{\mcnote}[1]{}
\newcommand{\dlnote}[1]{}
\newcommand{\ssnote}[1]{}
\newcommand{\pknote}[1]{}
\newcommand{\problem}[1]{}
\fi

\newcommand{\ack}{\textit{ack}}
\newcommand{\nack}{\textit{nack}}

\newcommand{\ignore}[1]{}

\newtheorem{theorem}{Theorem}

\newtheorem{definition}[theorem]{Definition}

\newenvironment{proof}[1][Proof]{\noindent\textbf{#1.} }{\hfill $\Box$\\[2mm]}

\begin{document}
\sloppy

%\title{Distributed Network Programming}
%\title{The Distributed SDN Update Problem:\\Towards Software Transactional Networks}

%\title{On Consistent Updates in Software Defined Networks under Unreliable Control}
%\title{The Case for Reliable Software Transactional Networking}

%\title{A Distributed SDN Control Plane for Concurrent Policy Updates}

\title{A Distributed SDN Control Plane for Consistent Policy Updates}

\author{
	Marco Canini$^{1}$ \quad Petr Kuznetsov$^{2}$ \quad Dan
        Levin$^{3}$ \quad Stefan Schmid$^{4}$\\
%\thanks{Contact author:
 %         stefan@net.t-labs.tu-berlin.de, FG INET, TEL 16,
  %        Ernst-Reuter Platz 7, 10587 Berlin, Germany}
\\
        $^{1}$ Universit\'{e} catholique de Louvain, Place Sainte Barbe 2, 1348 Louvain-la-Neuve, Belgium\\
        marco.canini@uclouvain.be\\
\\
        $^{2}$ T\'el\'ecom ParisTech, 46 Rue Barrault, 75013 Paris, France\\
        petr.kuznetsov@telecom-paristech.fr\\
\\
        $^{3}$ TU Berlin, Marchstr. 23, 10587 Berlin, Germany\\
        dlevin@inet.tu-berlin.de\\
\\
        $^{4}$ TU Berlin \& T-Labs, Ernst-Reuter Platz 7, 10587 Berlin, Germany\\
	    stefan.schmid@tu-berlin.de}
%	}

%\institute{}

\date{}

\maketitle

\thispagestyle{empty}

%\if \SAVESPACE 1
%\setlength{\floatsep}{3pt}
%\setlength{\textfloatsep}{3pt}
%\setlength{\dbltextfloatsep}{3pt}
%\setlength{\intextsep}{3pt}
%\setlength{\abovecaptionskip}{3pt}
%\fi

%\category{C.2.1}{Network Architecture and Design}{Centralized Networks}
%\category{C.2.4}{Distributed Systems}{Network Operating Systems}
%\terms{Measurement, Performance}
%\keywords{}

\begin{abstract}
Software-defined networking (SDN) is a novel paradigm that out-sources the
control of packet-forwarding switches to a set of software controllers.
The most fundamental task of these controllers is the correct implementation of
the \emph{network policy}, \ie, the intended network behavior. In essence,
such a policy specifies the rules by which packets must be forwarded across the
network.

This paper
%initiates the theoretical study of
studies
a distributed SDN control plane that enables \emph{concurrent} and \emph{robust} policy implementation.
We introduce a formal model describing the
interaction between the data plane and a distributed control plane
(consisting of a collection of fault-prone controllers).
Then we formulate the problem of \emph{consistent} composition of
concurrent network policy updates (short: the \emph{CPC Problem}).
To anticipate scenarios in  which some conflicting policy updates must
be rejected, we enable the composition via a natural \emph{transactional} interface with all-or-nothing semantics.

We show that the ability of an $f$-resilient distributed control
plane to process concurrent policy updates depends on the tag complexity,
\ie, the number of policy labels (a.k.a.~\emph{tags}) available to the controllers, and
describe a CPC protocol with optimal tag
complexity $f+2$.
\end{abstract}

%\vspace{1cm}

%\begin{center}
%{\bf [Regular paper only]}
%\end{center}
%[[PK I guess not anymore?]]

%\vspace{1cm}

%\begin{center}
%{\emph{Contact Address:}\\Marco Canini, Place Sainte Barbe~2, 1348 Louvain-la-Neuve, Belgium\\Tel: $+$32 10 47 48 32,
%marco.canini@uclouvain.be}
% {\emph{Contact Address:}\\Stefan Schmid, MAR 4-4, Marchstr.~23, 10587 Berlin, Germany\\Tel: $+$49 175 930 98 75,
% stefan.schmid@tu-berlin.de}
%\end{center}

%\newpage

%\begin{center}
%{\bf Regular and student paper: Dan Levin is a full-time student.}
%\end{center}

%\newpage
%\pagenumbering{arabic}\setcounter{page}{1}

%\section*{todo before submission}

\begin{comment}

\section*{to dos}

\ssnote{big open questions: tight bound, generalize algos for non-commutative, generalize algos for roger model}

\ssnote{
clarify that a policy may affect multiple ingress ports, and all-or-nothing semantics depends on that; enumerate tags
starting with 1 (\ie, 1,2,3,...).}

\pknote{introduction sounds too forwarding-oriented, shouldn't we
  rather talk general in the beginning?}

\end{comment}

\section{Introduction}

The emerging paradigm of Software-Defined Networking (SDN) promises to
simplify network management and enable building networks
that meet specific, end-to-end requirements.
%SDN centralizes on
In SDN, the \emph{control plane}  (a collection of network-attached servers)
maintains control over the so-called \emph{data plane} (the
packet-forwarding functionality implemented on switching
hardware). Control applications operate on a global, logically-centralized network view,
which introduces opportunities for network-wide management and optimization.
This view enables simplified programming models to define a high-level
network policy, \ie, the intended operational behavior of the network
encoded as a collection of \emph{forwarding rules} that the data plane must respect.

While the notion of centralized control lies at the heart of SDN,
%[[PK so wordy...
%a centralized controller
%constitutes one of the main challenges of the paradigm,
%as the controller constitutes a single-point of failure. Indeed a fully
%centralized system may not adequately or cost-effectively provide the
%required
implementing it on a centralized controller does not provide the required
%]]
levels of availability, responsiveness and scalability.
How to realize a robust, distributed control plane is one of the main open problems in
SDN
%[[PK which hard challenges?
%and is faced with the hard challenges and
and to solve it we must deal with
%]]
fundamental trade-offs between
different consistency models, system availability and performance.
%[[PK phrasing...
%
%The question of how to design a resilient control plane is essentially a
%distributed-computing one that requires reasoning about the behavior of a
%two-tier, distribute systems where control and data nodes have different
%processing capabilities. Indeed,
Implementing a resilient control plane becomes therefore a
distributed-computing problem that requires reasoning about
interactions and concurrency between the controllers while preserving
correct operation of the data plane.
%]]

In this paper, as a case study, we consider the problem of consistent
installation of network-policy \emph{updates} (\ie, collections of
state modifications spanning one or more switches),
one of the main tasks any network control
plane must support.
%[[PK: said below
%The policy updates must be installed on switching hardware respecting a property known
%as \emph{per-packet consistency}~\cite{network-update}.
%[[PK why in particular?
%In particular,
%]]
We consider a multi-authorship setting~\cite{Ferguson2013} where multiple
administrators, control applications, or end-host
applications may want to modify the network policy
independently at the same time, and where a conflict-free installation must be
found.

We assume that we are provided with a procedure
to assemble sequentially arriving policy updates in one  (semantically
sound) \emph{composed} policy (\eg, using the formalism of~\cite{netkat}).
Therefore, we address here the challenge of composing
\emph{concurrent} updates, while preserving a property known
as \emph{per-packet consistency}~\cite{network-update}.
Informally, we must guarantee that every packet traversing the network must be processed by
exactly one global network policy, even
throughout the interval during which the policy is updated --- in this case,
each packet is processed either using the policy in place prior to the update,
or the policy in place after the update completes, but never a mixture of the
two.
At the same time, we need to resolve conflicts among policy updates that
cannot be composed in a sequential execution. We do this by allowing
some of the requests to be \emph{rejected}, but requiring that
no data packet is affected by a rejected update.

%\smartparagraph{Our Contributions.}
%[[PK already said
%This paper initiates the study of the design of a robust SDN control plane
%from a distributed computing perspective.
Our first contribution is a formal model
of SDN under fault-prone, concurrent control.
%]]
%[[PK already said
%As a case study,
We then focus on the problem of
%]]
\emph{per-packet consistent updates}~\cite{network-update}, and introduce the
abstraction of \emph{Consistent Policy Composition (CPC)}, which offers a
\emph{transactional} interface to address the issue of conflicting policy
updates. We believe that the CPC abstraction, inspired by the popular paradigm
of software transactional memory (STM)~\cite{stm-st95}, exactly matches the
desired behavior from the network operator's perspective, since it captures
the intuition of a correct sequential composition combined with optimistic
application of policy updates.

%[[PK recall Tag/Path
We then discuss different protocols to solve the CPC problem.
We present a \emph{wait-free} CPC algorithm, called {\CPO},
%[[PK ingress ports never mentioned, also is this the main characteristics
%which orders policies at the ingress ports.
%]]
which allows the controllers to directly apply their updates on the
data plane and resolve conflicts as they progress installing the updates.
While {\CPO} tolerates any number of faulty
controllers 
%[[PK because it is wait-free
%\mcnote{why any number?} 
%]]
and does not require them to be strongly synchronized
(thus improving concurrency of updates),
it incurs a linear \emph{tag complexity} in the number of to-be-installed policies (and hence in the worst-case
 exponential in the network size).
%We introduce the notion of \emph{tag complexity} to
%evaluate the performance of a protocol, and
We then present a more sophisticated protocol called {\DPO}, which
applies the replicated state-machine approach to implement
a total order on to-be-installed 
%\mcnote{why installed? isn't it to-be-installed?} 
policy updates.
Assuming that at most $f$ controllers can fail, we show that {\DPO}
achieves an optimal tag complexity $f+2$.

To the best of our knowledge, this work initiates
an analytical study of a \emph{distributed} and \emph{fault-tolerant} SDN control plane.
We keep our model intentionally simple
and we  focus on a restricted class of forwarding
policies, which was sufficient to highlight
% plane from
%a distributed-computing perspective. Generally,
intriguing connections between our SDN model and conventional
distributed-computing models, in particular, STM~\cite{stm-st95}.
One can view the SDN data plane as a shared-memory data
structure, and the controllers can be seen as read/write processes,
modifying the forwarding rules applied to packets at each switch.
The traces of packets constituting the data-plane workload can be seen as
``read-only'' transactions, reading the
forwarding rules at a certain switch in order to ``decide'' which switch state to read
next.
Interestingly, since in-flight packets cannot be dropped (if it is not intended to do so) or delayed, these
read-only transactions must always commit, in contrast with policy update transactions.

%[[PK phrasing
In general, we believe that our work can inform the networking community about what can and
cannot be achieved in a distributed control plane.
We also derive a minimal requirement on the SDN model without which CPC is impossible to
solve.
%]]
%[[PK computing?
%For the networking community,
From the distributed-computing perspective,
%our work opens an interesting field: as
we show that the SDN model exhibits concurrency phenomena not
yet observed in classical distributed systems.
For example, even if the controllers can synchronize their actions using
consensus~\cite{Her91},  complex interleavings between the
controllers' actions and packet-processing events prevent them from implementing CPC with
constant tag complexity (achievable using one reliable controller).
%comes with several new characteristics which are not
%well-understood yet from a distributed computing perspective.
%]]

\smartparagraph{Roadmap.} We introduce our
SDN model in Section~\ref{sec:model}. Section~\ref{sec:problem} formulates the CPC problem
and Section~\ref{sec:solution} describes our CPC solutions and their
complexity bounds. We discuss related work in
Section~\ref{sec:related} and conclude in Section~\ref{sec:conc}.
Proof sketches are given in the Appendix.

\section{Distributed Control Plane Model}\label{sec:model}

We consider a setting where different users (\ie, policy authors or administrators) can issue
policy update requests to the distributed SDN control plane.
%Each request will be received by (resp.~replicated to) any non-faulty process. \mcnote{What does it mean to say: resp.~replicated to?}
We now introduce our SDN model as well as the policy concept in more detail.

\vspace{1mm}\noindent\textbf{Control plane.}
The distributed \emph{control plane} is modelled as a set of $n\geq 2$ \emph{controllers}, $p_1,\ldots,p_n$.
The controllers are subject to \emph{crash} failures: a
faulty controller stops taking steps of its algorithm.
The controller that never crashes is called \emph{correct} and we assume
that there is at least one correct controller.
We assume that controllers can communicate among themselves (\eg, through an out-of-band management network) in a reliable but asynchronous
(and not necessarily FIFO) fashion, using message-passing. Moreover,
the controllers have access to a consensus
abstraction~\cite{FLP85} that allows them to implement, in a fault-tolerant manner,
any replicated state machine, provided its sequential
specification~\cite{Her91}.
The consensus abstraction can be obtained, \eg, assuming
the eventually synchronous communication~\cite{DDS87} or
the \emph{eventual leader} $\Omega$ failure detector~\cite{CHT96}
shared by the controllers, assuming a majority of correct controllers
or the \emph{quorum} failure detector $\Sigma$~\cite{DFG10}.
%\ssnote{Say something how this consensus abstraction can be realized
%in practice?}

\vspace{1mm}\noindent\textbf{Data plane.}
Following~\cite{network-update}, we model the \emph{network data plane}
%as a packet processor: a \emph{network} is modeled
as a set $P$ of \emph{ports} and a set $L\subseteq P\times P$ of directed
\emph{links}.
%A link $(i,j)\in L$ is called \emph{outgoing} for port $i$ and
%\emph{incoming} for port $j$.
As in~\cite{network-update}, a hardware switch is represented as
  a set of ports,  and a physical bi-directional link between two
  switches $A$ and $B$ is represented as a set of \emph{directional}
  links,  where each port of $A$ is connected to the port of $B$
  facing $A$ and every port of $B$ is connected to the port of $A$
  facing $B$.
%Let $L(i)$, $i\in P$, denote the set of \emph{neighbors} of $i$, \ie,
%$\{j \in P | (i,j)\in L\}$.
We additionally assume that $P$ contains two distinct ports,
$\textsf{World}$ and $\textsf{Drop}$,     
which represent forwarding a packet to the outside of the network (\eg, to an end-host or upstream provider) and
dropping the packet, respectively. 
A port $i\notin\{\textsf{World},\textsf{Drop}\}$ that has no \emph{incoming} links, \ie, $\nexists j\in P$:
$(j,i)\in L$  is called \emph{ingress}, 
otherwise the port is called \emph{internal}.
Every internal port is connected to $\textsf{Drop}$ (can drop
packets). A subset of ports are connected to $\textsf{Drop}$ (can
forward packets to the outside of the network).
$\textsf{World}$ and $\textsf{Drop}$ have no outgoing links: $\forall
i\in \{\textsf{World},\textsf{Drop}\},\;\nexists j\in P$:
$(i,j)\in L$.
%\mcnote{What is World for? What is connected to World? Is it every ingress port connected to World?}

The workload on the data plane consists of a set $\Pi$ of \emph{packets}.
(To distinguish control-plane from data-plane communication,
we reserve the term \emph{message} for a communication involving at least one controller.)
%and the term \emph{packet} for data plane communication.
In general, we will use the term \emph{packet} canonically as a type~\cite{network-update},
\eg, describing all packets (the packet \emph{instances} or \emph{copies}) matching a certain header;
when clear from the context, we do not explicitly distinguish between
packet types and packet instances.
%\pknote{Is this type distinction important?} \ssnote{I think makes sense and we should keep it}
%
%For instance, a policy usually refers to packet \emph{types} whereas per-packet consistency refers to packet \emph{instances}.
%\pknote{unclear...}

\vspace{1mm}\noindent\textbf{Port queues and switch functions.}
The \emph{state} of the network is characterized by a \emph{port
  queue} $Q_i$ and a \emph{switch function} $S_i$ associated with
every port $i$.
A port queue $Q_i$ is a sequence of packets that are, intuitively, waiting to be processed at port $i$.
%\mcnote{The definition of switch function is not very clear to me. I would think that it reads from a queue, processes the packet and places it to another queue. Also, located packet seems like an important concept but it is only introduced in an indirect way.}
A switch function is a map $S_i:\;\Pi\rightarrow \Pi\times P$,
that, intuitively, defines how packets in
the port queue $Q_i$ are to be processed.
When a packet $\textit{pk}$ is fetched from port queue $Q_i$, the corresponding \emph{located
  packet}, \ie, a pair $(\textit{pk}',j)=S_i(\textit{pk})$ is computed
and the packet $\textit{pk}'$ is placed to the queue $Q_j$.

We represent the switch function at port $i$, $S_i$, as a collection of
\emph{rules}. Here a rule $r$ is a partial map $r:\Pi\rightarrow \Pi\times P$
that, for each packet $pk$ in its domain $\textit{dom}(r)$, generates a new
located packet $r(pk)=(pk',j)$, which results in $pk'$ put in queue $Q_j$ such
that $(i,j)\in L$. Disambiguation between rules that have overlapping domains
is achieved through priority levels, as discussed below.
%[[PK sounds good to me, commented for the 1st submission
%\mcnote{Marco to revise this.}
%

%[[PK not really needed?
%We assume that every rule $r\in S_i$ is \emph{applicable to a port $i$}, \ie,
%$r$ only puts packets to the queues of the ports directly connected to $i$:
%$\forall pk\in\textit{dom}(r)$ and $(pk',j)=r(pk)$, $(i,j)\in L$.
%We assume that each port is configured to process every
%packet: $\bigcup_{r\in S_i} \dom(r) =\Pi$.
%]]

We assume that only a part of a packet $pk$ can be modified by a rule: namely,
a header field called the \emph{tag} that carries the information that is used
to identify which rules apply to a given packet.
%[[PK never used
%We denote the tag of a packet $pk$ by $\tau(pk)$
%and denote by
%$|\tau|$ the number of bits used for tags.
%]]
%[[PK dropped for the moment, probably not needed
%For simplicity, we assume that only ingress port may contain rules
%modifying the packets. The internal ports are only allowed to forward
%packets without modifying them.
%]]

%\pknote{the notation $|\tau|$ is strange}

\vspace{1mm}\noindent\textbf{Port operations.}
We assume that a port supports an \emph{atomic} execution of a \textit{read}, \textit{modify-rule}
and \textit{write} operation: the rules of a port can be atomically read and, depending
on the read rules, modified and written back to the port.
%We will refer to such a port as an \emph{\textit{rmw}-port}.
Formally, a port $i$ supports the operation:
$\textit{update}(i,g)$, where $g$ is a function defined on the
sets of rules.
The operation atomically reads the state of the port, and then, depending
on the state, uses $g$ to update it and return a response.
For example, $g$ may involve adding a new forwarding rule or
a rule that puts a new tag $\tau$ into the headers
of all incoming packets.

%[[PK
\ignore{
In Section~\ref{sec:discussion}, we give a valency proof that
under a weaker port model of atomic reads and writes,
the CPC Problem is impossible to solve.
}

%We assume for any set of rules $r,r'\in S_i$ installed
%on the port, there is a deterministic way of resolving internal conflicts.
%we have $\dom(r)\cap\dom(r')=\emptyset$, or
%$\dom(r)\subsetneq\dom(r')$, or $\dom(r')\subsetneq\dom(r)$.

%As we define formally below,
%a packet arriving at a port is processed
%according to the highest priority rule that applies to the packet
%(includes the packet in its domain).
%Therefore, any update of the port's state is taking effect only if the
%resulting set of rules is internally consistent.
%Respectively, we assume that in the initial state of the network, the state of every
%port is internally consistent.

\vspace{1mm}\noindent\textbf{Policies and policy composition.}
Finally we are ready to define the fundamental notion of network policy.
A \emph{policy} $\pi$ is defined by a \emph{domain} $\dom(\pi)\subseteq
\Pi$, a \emph{priority level} $\pr(\pi) \in\mathbb{N}$, and a unique
\emph{forwarding path}, \ie, a loop-free sequence of piecewise connected ports,
for each ingress port that should apply to the packets in its domain
$\textit{dom}(\pi)$. Formally, for each ingress port $i$ and each packet
$pk\in\dom(\pi)$ arriving at port $i$, $\pi$ specifies a sequence of distinct
ports $i_1,\ldots,i_s$ that $pk$ should follow, where $i_1=i$, $\forall
j=1,\ldots,s-1$, $(i_j,i_{j+1})\in L$ and
$i_s\in\{\textsf{World},\textsf{Drop}\}$.
%
%[[PK confusing at this point
%$\CPOs(\pi)=\{\pi^{(1)},\ldots,\pi^{(s)}\}$,
%]]
%Throughout this paper,
The last condition means that each packet following the path
eventually leaves the network or is dropped.
%[[PK redundant?
%Also, we assume that policy paths are \emph{loop-free}.
%]]

We call two policies $\pi$ and $\pi'$ \emph{independent} if
$\dom(\pi)\cap \dom(\pi')= \emptyset$.
Two policies $\pi$ and $\pi'$ \emph{conflict} if they are not
independent and $\pr(\pi)=\pr(\pi')$.
%[[PK not used?
%We say that a policy $\pi$ \emph{refines} policy $\pi'$ if
%$\dom(\pi)\subseteq \dom(\pi')$.
%]]
%
Now a set $U$ of policies is \emph{conflict-free} if
no two policies in $U$ conflict.
%$\pi,\pi'\in S$, $\pi$ and $\pi'$ are independent, or if $|\pr(\pi)-\pr(\pi')|>0$,
%or $|\pr(\pi)-\pr(\pi')|=0$
%and $\pi$ refines $\pi'$ or vice versa.
Intuitively, the priority levels are used to establish the order among
non-conflicting policies with overlapping domains: a packet
$pk\in\dom(\pi)\cap \dom(\pi')$, where $\pr(\pi)>\pr(\pi')$,  is
processed by policy $\pi$.
Conflict-free policies in a set $U$ can therefore be \emph{composed}:
a packet arriving at a port is applied the highest priority policy
$\pi\in U$ such that
$\textit{pk}\in\dom(\pi)$.
%[[PK the footnote does not help much imho.
%\footnote{For pedagogical reasons, we chose the
%  sequential composition procedure to be as simple as possible.
%However, our algorithms work for arbitrary policy compositions
%as long as concurrent policy updates commute.~\cite{frenetic}}
%]]

%[[PK matters for weaker port model
%We assume that if $S_i$ is not conflict-free and the packet is subject
%to two conflicting rules, then the port is
%configured to apply some default deterministic rule to
%select one of these rules to be applied to the packet. \ssnote{Do we need this?}
%]]

\vspace{1mm}\noindent\textbf{Modelling traffic.}
The traffic workload on our system is modelled using \textit{inject}
and \textit{forward} events defined as follows:
\begin{itemize}
\item $\textit{inject}(pk,j)$: the environment injects a packet $pk$
  to an ingress port $j$ by adding $pk$ to the end of queue $Q_j$,
  \ie, replacing $Q_j$ with $Q_j\cdot pk$.
\item $\textit{forward}(pk,j,pk',k)$, $j \in P$:
%[[PK resolved now
%\mcnote{why say $\notin\{\textsf{World},\textsf{Drop}$? btw, are these in $P$?}
%]]
  the first packet in $Q_j$ is processed
  according to $S_j$, \ie, if $Q_j=pk.Q'$, then $Q_j$ is
  replaced with $Q'$ and  $Q_k$ is
  replaced with $Q_k\cdot pk'$, where $r(pk)=(pk',k)$ and $r$ is the
  highest-priority rule in $S_j$ that can be applied to $pk$.
%Recall that if there is a conflict, \ie, two different rules
%applicable to $pk$ have the highest priority, then the conflict is
%resolved using some deterministic procedure. \ssnote{Do we need this?}
\end{itemize}

%\vspace{1mm}
\noindent\textbf{Algorithms, histories,  and problems.}
Each controller $p_i$ is assigned with an \emph{algorithm}, \ie, a state machine that
$(i)$ accepts invocations of high-level operations, $(ii)$ accesses ports with
\textit{read-modify-write} operations,
%\textit{read},
%\textit{add-rule} or \textit{remove-rule} operations,
$(iii)$ communicates with other controllers, and $(iv)$ produces high-level responses.
%A \emph{problem} is a set $\mathcal H$ of desirable histories.
The distributed algorithm generates a sequence of \emph{executions}
consisting of port accesses, invocations, responses, and packet forward events.
Given an execution of an algorithm, a \emph{history} is the sequence of externally observable events, \ie, $\textit{inject}$
and $\textit{forward}$ events, as well as
invocations and responses of controllers' operations.

We assume an asynchronous \emph{fair} scheduler and \emph{reliable} communication channels
between the controllers: in every infinite execution, no message starves in a port queue without being
served by a \emph{forward} event, and
every message sent to a controller is eventually received.
%\mcnote{Forward only applies to data plane packets.}

A \emph{problem} is a set $\mathcal P$ of histories.
An algorithm solves a problem $\mathcal P$ if the history of its every
execution is in $\mathcal P$.
An algorithm solves $\mathcal P$ \emph{$f$-resiliently} if the property above holds in
every $f$-resilient execution, \ie, in which at most $f$ controllers
take only finitely many steps.
An $(n-1)$-resilient solution is sometimes called \emph{wait-free}.

\vspace{1mm}\noindent\textbf{Traces and trace consistency.}
In a history $H$, every packet injected to the network
generates a \emph{trace}, \ie, a sequence of located packets: each event
$ev=\textit{inject}(pk,j)$ in $E$ results in $(pk,j)$ as the first element of
the sequence, $\textit{forward}(pk,j,pk_1,k_1)$ 
%\mcnote{What is this? We defined it as
%$\textit{forward}(pk,j,pk',k)$} 
adds $(pk_1,j_1)$ to the trace, 
%\mcnote{unclear what $(pk_1,j_1)$ is}, 
and each next $\textit{forward}(pk_k,j_k,pk_{k+1},j_{k+1})$ extends
the trace with $(pk_{k+1},j_{k+1})$, unless
$j_k\in\{textsf{Drop},\textsf{World}\}$ in which case we say that the
trace \emph{terminates}.
Note that in a finite network an infinite trace must contain a cycle.
%[[PK yes, but a trace may not respect any policy in general
%\mcnote{Not sure what this sentence about cycle adds. We already said we assume paths do not have cycles.}
%]]
%
Let $\rho_{ev,H}$
%(resp. $\rho_{ev,E}$ )
denote the trace corresponding
to an inject event $ev=\textit{inject}(pk,j)$ a history $H$.  %in an execution $E$.
Trace $\rho=(pk_1,i_1),(pk_2,i_2),\ldots$ is
\emph{consistent with a policy $\pi$} if
$pk_1\in \dom(\pi)$ and $(i_1,,i_2,\ldots) \in \pi$.
%[[PK decided to keep it: it refers to per-packet consistency at the
%end, I think it's ok
%\mcnote{Replace consistent with A trace respects a policy.}
%]]

\vspace{1mm}\noindent\textbf{Tag complexity.} It turns out that what can and what cannot be achieved
by a distributed control plane depends on the number of available tags,
used by control protocols to distinguish packets that should be processed by different policies.
%[[PK commented for the moment
%\mcnote{It is very odd that tag has not been introduced and this talks about tag complexity.
%We need to explain the role of the tag in the policy description.}
%]]
Throughout this
paper, we will refer to the number of different tags used by a protocol as the
\emph{tag complexity}. W.l.o.g., we will typically assume that tags are integers $\{0,1,2,\ldots\}$,
and our protocols seek to choose low tags first; thus, the tag complexity is usually the largest
used tag number $x$, throughout the entire (possibly infinite) execution of the protocol and in the
worst case. Observe that a protocol of tag complexity $x$ requires $\lfloor\log{x}\rfloor+1$ bits
in the packet header.
%[[PK ne need really, in fact bits are never mentioned again
%; accordingly, in this paper, we will sometimes refer to the \emph{number of (tag) bits}
%needed by a protocol.
%\mcnote{Why do we care about the number of bits?}
%]]

\vspace{1mm}\noindent\textbf{Monitoring oracle.} In order to be able to reuse tags, the control plane needs some feedback
from the network about the \emph{active policies}, \ie, for which policies there are still packets in transit.
We use an oracle model in this paper: each controller can query the oracle to learn about the tags currently
in use by packets in any queue. Our assumptions on the oracle are minimal, and oracle interactions can be asynchronous.
In practice, the available tags can simply be estimated by assuming a rough upper bound on the transit time of
packets through the network.

%============================================================
\section{The CPC Problem}\label{sec:problem}
%============================================================

Now we formulate our problem statement.
At a high level, the CPC abstraction of consistent policy composition
accepts concurrent \emph{policy-update requests} and makes sure that
the requests affect the traffic as a \emph{sequential composition} of
their policies.
The abstraction offers a transactional interface
where requests can be \emph{committed} or \emph{aborted}.
Intuitively, once a request commits, the corresponding policy
affects every packet in its
domain that is subsequently injected.
%\mcnote{wait. isn't it every packet that enters the network after the request is committed?}
But in case it cannot be composed with the currently
installed policy, it is \emph{aborted} and does not affect a
single packet. On the progress side, we require that if a set of policies conflict, at least one
policy is successfully installed.
%\ssnote{How is this property called officially?}
%\ssnote{Shall we say that
%for now, we assume policy updates commute? if so, we may also say that policies will ``last forever'', as
%they cannot be removed (only overridden with higher priority?)}
%\pknote{I'd put it as a possible extension}
%[[PK It is not part of the pb statement, it's part of the
%model. Besides ``immediately'' sounds strange
%We require that each packet arriving at a port is forwarded \emph{immediately};
Recall that inject and forward events are not under our control, \ie, the packets cannot
be delayed, \eg, until a certain policy is installed.
%]]
%[[PK need to say at some point what we mean by per-packet consistency
Therefore, a packet trace that interleaves with a policy update
must be consistent with the policy \emph{before} the update or the policy \emph{after}
the update (and not some partial policy), the property is
referred to as \emph{per-packet consistency}~\cite{network-update}). 
%]]
%%policies, either completely including the concurrent update or
%%completely excluding it  
%(the property referred to as \emph{per-packet consistency}~\cite{network-update}). 

\vspace{1mm}\noindent\textbf{CPC Interface.}
Formally, every controller $p_i$ accepts requests
$\textit{apply}_i(\pi)$, where $\pi$ is a policy,
and returns $\ack_{i}$ (the request is committed) or $\nack_{i}$
(the request is aborted).

We specify a partial order relation on the events in a history $H$,
denoted $<_H$.
We say that a request $\textit{req}$ \emph{precedes} a request
$\textit{req}'$ in a history $H$, and we write $\textit{req} <_H
\textit{req}'$, if the response of
$\textit{req}$ appears before the invocation of $\textit{req}'$ in $H$.
If none of the requests precedes the other, we say that the requests
are \emph{concurrent}.
Similarly, we say that an inject event $ev$ \emph{precedes} (resp.,
\emph{succeeds}) a request $\textit{req}$ in $H$, and we write $ev <_H
\textit{req}$ (resp., $\textit{req} <_H ev$),  if $ev$ appears after the
response (resp., before  the invocation) of
$\textit{req}$ in $H$.
Two inject events $ev$ and $ev'$ on the same port in $H$
are related by $ev <_H ev'$ if
$ev$ precedes $ev'$ in $H$.

An inject event $ev$ is concurrent with
$\textit{req}$ if $ev \not<_H \textit{req}$ and $\textit{req} \not<_H ev$.
A history $H$ is \emph{sequential} if in $H$, no two requests are
concurrent and no inject event is concurrent with a request.

Let $H|p_i$ denote the \emph{local} history
of controller $p_i$, \ie, the subsequence of $H$ consisting of all events
of $p_i$.
We assume that every controller is \emph{well-formed}:
every local history $H|p_i$ is sequential, \ie,
no controller accepts a new request before producing a response to the
previous one.
A request issued by $p_i$ is \emph{complete} in $H$ if it is followed by a
matching response ($\ack_i$ or $\nack_i$) in $H|p_i$ (otherwise it is
called \emph{incomplete}).
A history is \emph{complete} if every request is complete in $H$.
A \emph{completion} of a history $H$ is a complete history $H'$ which is like
$H$ except that  each incomplete request in $H$
is completed with $\ack$ (intuitively, this is necessary if the
request already affected packets) or $\nack$ inserted
somewhere after its invocation.
%[[PK why this?
%(We require that our implementation is correct
%also for infinite histories and even if a failed process does not return an ACK or NACK.)
%]]
%
Two histories $H$ and $H'$ are \emph{equivalent} if $H$ and $H'$ have
the same sets of events,  for all $p_i$, $H|p_i=H'|p_i$, and for all
inject events $ev$ in $H$ and $H'$, $\rho_{ev,H}=\rho_{ev,H'}$.

%For an inject event $ev$ in $H$, let $H|ev$ denote the subsequence of
%$H$ that only consists of the path $\rho_{ev,H}$.

\vspace{1mm}\noindent\textbf{Sequentially composable histories.}
A sequential complete history $H$ is \emph{legal} if the following two
properties are satisfied:
%(1) the set of committed policies in $H$ is conflict-free, and
%[[PK a new version
(1) a policy
is committed in $H$ if and only if
%[[PK must be if and only if for consistency ]]
it does not conflict with the set of
policies previously committed in $H$, and
%\ssnote{(2) a maximal subset of conflicting
%policies is installed (\ie conflicting policies cannot simply be
%rejected all)}, and (3)
(2) for every inject event
$ev=\textit{inject}(pk,j)$ in $H$,  the trace $\rho_{ev,H}$  is consistent
with the composition of all committed policies that precede
$ev$ in $H$.
%]]

%, \ie, the most specific policy $\pi$ committed before $ev$
%in $H$ such that $pk\in\dom(\pi)$.

\begin{definition}[Sequentially composable history]\label{def:history}
We say that a complete history $H$ is \emph{sequentially composable}
if there exists a legal sequential history $S$ such that (1) $H$ and
$S$ are equivalent, and (2)  $<_H\subseteq <_S$.
%\begin{itemize}
%\item $H$ and $S$ are equivalent, and
%\item $<_H\subseteq <_S$.
%\end{itemize}
\end{definition}
Intuitively, Definition~\ref{def:history} implies that the traffic in
$H$ is processed \emph{as if} the requests were applied atomically and
every injected packet is processed instantaneously (per-packet consistency).
The legality property here requires that only committed requests affect the traffic.
Moreover, the equivalent sequential history $S$ must respect the
order in which non-concurrent requests take place and packets arrive in $H$.

\begin{definition}[CPC]\label{def:CPC}
We say that an algorithm solves the problem of Consistent Policy
Composition (CPC) if for its every history $H$, there exists a
completion $H'$ such that:
\begin{description}

\item[Consistency.] $H'$ is sequentially composable.

%[[PK not needed anymore
%\item[Progress.] If a request $\textit{req}=\textit{apply}_i(\pi)$ completes with
%  $\nack_i$ in $H'$, then there exists  $\textit{req}'=\textit{apply}_j(\pi)$
%  in $H'$,  concurrent  or committed before the response of
%  $\textit{req}$ such that $\pi$ and $\pi'$ are conflicting
%  ($\{\pi,\pi'\}$ is not conflict-free). %\ssnote{would be nice to have at least one accepted?}
%]]

\item[Termination.]  If $H$ is infinite, then every correct controller $p_i$
  that accepts a requests $\textit{apply}_i(\pi)$, eventually returns ($\ack_i$ or
  $\nack_i$) in $H$.
\end{description}
\end{definition}
%[[PK
%\ssnote{Should we add the requirement that at least one out of conflicting policies are installed? Or does
%this belong to the legal history?}
%]]
%
%
%

\begin{figure}[t]
\centering
\subfigure[]{
	\includegraphics[width=.28\columnwidth]{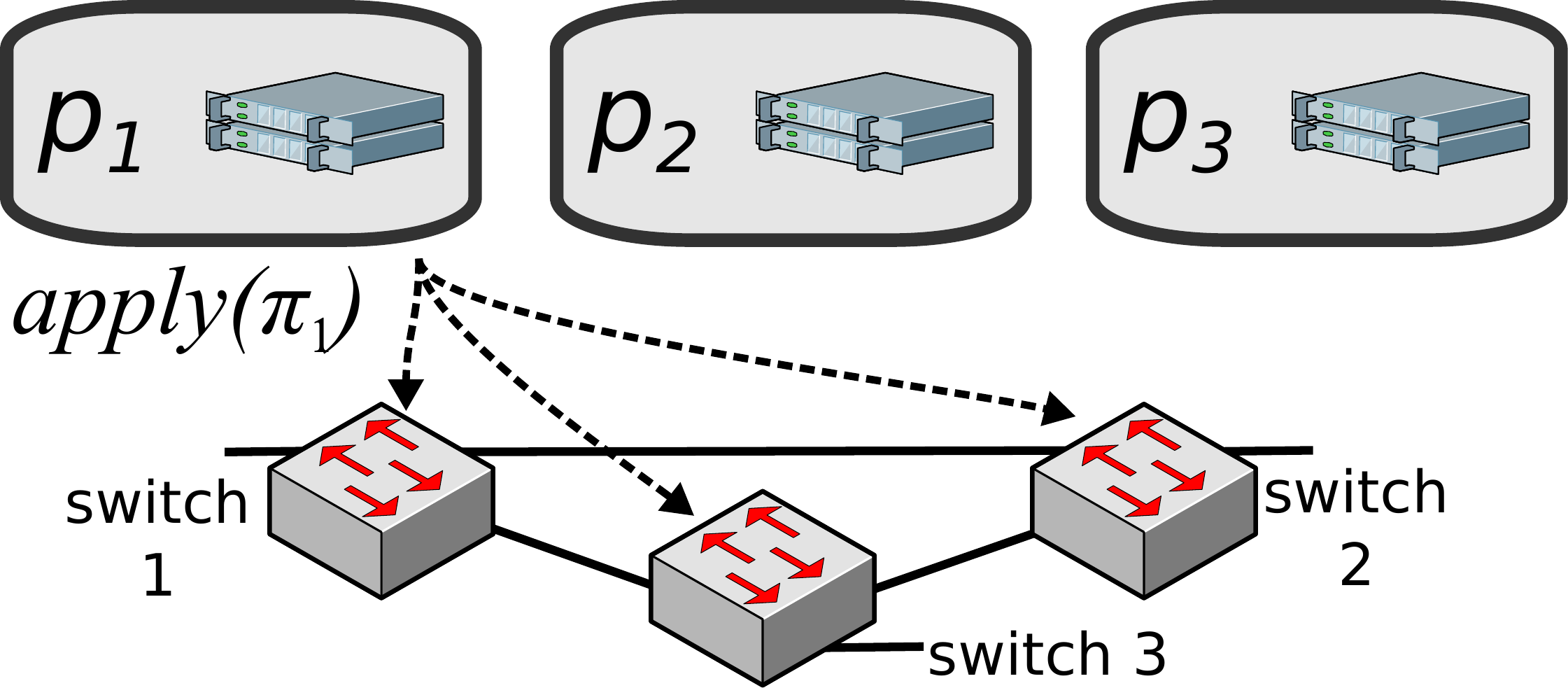}
	\label{fig:example1}
}
\subfigure[]{
	\includegraphics[width=.68\columnwidth]{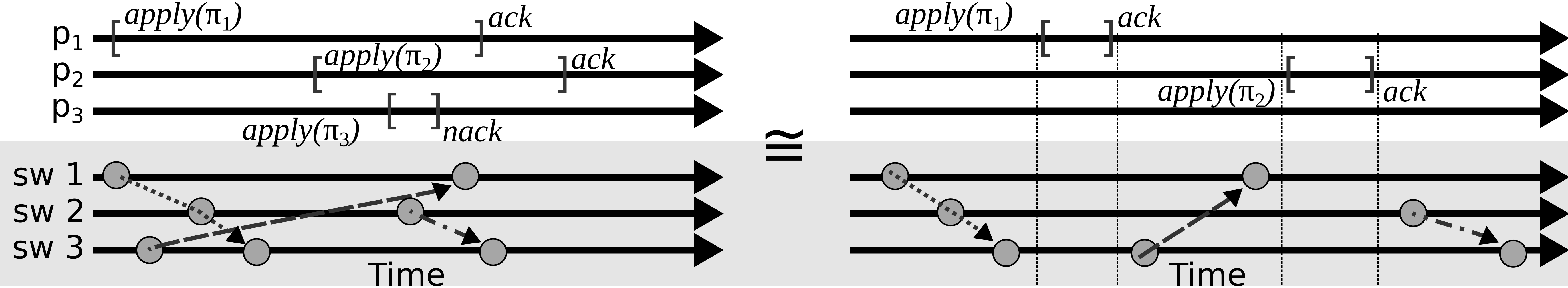}
	\label{fig:example2}
}
\vspace{-1.0em}
\caption{\small Example of a policy composition with a $3$-controller control
  plane and $3$-switch data plane (a).
%[[PK too wordy? redundant?
%The function of the controllers is to  accept policy-update requests issued by the control
%applications and try installing the updates on the switches.
%]]
The three controllers try to
concurrently install three different policies $\pi_1$, $\pi_2$, and
$\pi_3$. We suppose that  $\pi_3$ is conflicting with both
$\pi_1$ and $\pi_2$, so $\pi_3$ is aborted (b).
Circles represent data-plane events (an \emph{inject} event followed
by a sequence of forward events).
Next to $H$ we depict its ``sequential equivalent''
$H_S$. In the sequential history, no two requests
are applied concurrently and no request is rejected.}
\label{fig:examples}
\vspace{-1.0em}
\end{figure}

Note that, for an infinite history $H$, the Consistency and Termination
requirements imply that an incomplete request in $H$ can only cause
aborts of conflicting requests for a finite period of time: eventually it
would abort or commit in a completion of $H$ and if it aborts, then no
subsequent conflicting requests will be affected.
As a result we provide an all-or-nothing semantics:
a policy update, regardless of the behavior of the controller that
installs it,  either eventually takes effect or does not affect a single packet.
Figure~\ref{fig:examples} gives an example of a sequentially
composable history.
% (A more detailed
%example can be found in~\cite{STN-disc13}).

%============================================================
\section{CPC Solutions and Complexity Bounds}\label{sec:solution}
%============================================================

We now discuss how the CPC problem can be solved and analyze the
complexity its solutions incur. We begin with a simple wait-free
algorithm $\CPO$ which implicitly orders policies at a given ingress port;
$\CPO$ incurs a linear tag complexity in the number of to-be-installed policies.
Then we present an $f$-resilient algorithm $\DPO$ with
tag complexity $f+2$.
%, where $f$ is the upper bound on the number  of
%faulty controllers.
We also show that {\DPO} is optimal, \ie, no
protocol can maintain smaller tags for all networks.

\subsection{$\CPO$: Per-Policy Tags}\label{ssec:path-algo}

The basic idea of $\CPO$ is to encode each possible forwarding path in
the network by its own tag.
%that any policy may use, by its own tag.
Let $\tau_k$ be the tag representing the $k^{\mathit{th}}$ possible path.
$\CPO$  assumes that, initially, for each internal port $i_x$
%(\ie, for each port which is not an ingress port)
which lies on the $k^{\mathit{th}}$ path, a rule $r_{\tau_k}(pk)=(pk,
i_{x+1})$ is installed, which forwards \emph{any packet} tagged
$\tau_k$ and forwards the corresponding packet to the path's
successive port $i_{x+1}$.
%[[Pk confusing and not important
%Note that $\CPO$ does not modify the initial rules at the
%internal ports.
%\mcnote{What are the initial rules?}
%]]

Upon receiving a new policy request $\pi$ and before installing any rules,
a controller $p_i$ executing $\CPO$ sends a message to \emph{all}
other controllers
%[[PK
%$p_j$ ($j\neq i$)
%]]
informing them about the rules it intends to add to the
ingress ports; every controller receiving this message
rebroadcasts it (making the broadcast reliable), and starts installing the policy on $p_i$'s behalf.
This ensures that every policy update that started affecting the
traffic eventually completes.
% This \emph{policy message} is to ensure fault-tolerance: whenever a process
%$p_j$ receives such a message, it starts installing the policy $\pi$
%of $p_i$ as well  (in addition to its own policies).
%[[PK the intuition is not clear at this point
%Intuitively, as a result, every policy update that affects the traffic is
%e ventually completed and serialized in an equivalent legal sequential history.
%]]
Let $i_1,\ldots,i_s$ be the set of ingress port, and $\pi^j$ be the
path specified by policy $\pi$ for ingress port $i_j$, $j=1,\ldots,s$.
%and let $i_1,\ldots,i_s$ be the respective ingress ports of the paths
%$\pi^{(1)},\ldots,\pi^{(s)}$.
To install $\pi$, $\CPO$ seeks to add a rule to each ingress port $i_j$; this rule tags all
packets matching the policy domain with the tag describing the path $\pi^j$.
However, since different policies from different controllers
may conflict, %(\eg, have different priorities), [[PK different?]]
every controller updates the ingress ports in a pre-defined order.
%[[PK shorter, $\prec$ is only used in the proof
%imposes a strict order $\prec$
%on all ingress ports: all processes will update the ingress ports (using atomic \textit{update}) one-by-one according
%to $\prec$.
%]]
Thus, conflicts are discovered already at the lowest-order port, and the conflict-free all-or-nothing installation of a policy
is ensured.
%[[PK not important
%(The order $\prec$ may be based on switch or port identifiers.)
%]]

%\ignore{
%\begin{proof}
%Let us first describe $\CPO$ more formally. Let $\tau_k$ be the unique tag of the $k^{\mathit{th}}$ path. First, for each internal port $i_x$ (\ie, for each %port which is not an ingress port)
%which lies on the $k^{\mathit{th}}$ path, $\CPO$ installs a rule $r_{\tau_m}(pk)=(pk, i_{x+1})$ which matches \emph{any packet} with this tag
%$\tau_k$, and forwards the corresponding packets to the port $i_{x+1}$ which is given by path.
%Subsequently and before installing its policy $\pi$, a process $p_i$ executing $\CPO$ also informs \emph{all} other processes $p_j$ ($j\neq i$) about
%the rules it intends to add to the ingress ports. Whenever a process $p_j$ receives such a message, it starts installing the policy $\pi$ of $p_i$ as well (in addition %to its own policies).
%Process $p_i$ then iterates over all ingress ports $p$ according to the fixed order $\prec$, and installs the following rule using a \textit{rmw} operation: for each %packet $pk$ matching the domain of the policy, the path tag $\tau_i$ is attached
%and forwarded to the first port $j$, \ie, $r(pk)=(\tau(pk) \leftarrow \tau_i,j)$. In case the ingress port $p$ already contains a rule of a conflicting or equivalent %policy, the
%rule is not installed.
%%
%
%\hfill $\Box$ \end{proof}
%}
%
Observe that $\CPO$ does not require \emph{any} feedback from
the network on when packets arrive or leave the system. It just tags
all traffic at the network edge;
%\footnote{In this light,
%  $\CPO$ is similar to the Multiprotocol Label Switching
%  (MPLS) protocol frequently used for traffic engineering in the
%  Internet.\pknote{do we need this footnote?}}
internally, the packets are only forwarded
according to these tags.
%, and no other header or payload bits are matched.

We have the following theorem.
\begin{theorem}\label{thm:tags}
$\CPO$ solves the CPC problem in the wait-free manner, without
  relying on the oracle and consensus objects.
%In the \textit{rmw}-port model, the CPC problem can be solved $(n-1)$-resiliently.
%even for $n-1$ failures where $n$ is the number of processes.
\end{theorem}
However, while providing a correct network update even under high
control plane concurrency and failures, $\CPO$ has a large
tag complexity, namely linear in the number of to-be-installed
policies (which may grow to super-exponential in the network size).
If we want to reduce the tag overhead, we should be able to
\emph{reuse} tags that are not needed anymore.
%[[PK this is not the case with our algorithms, it is rather
%%  proportional to the number of failures, but it's hard to motivate 
%Ideally, the number of tags should only depend on the number 
%of concurrently installed policies.
%]]

\subsection{$\DPO$: Optimal Tag Complexity}

The $\DPO$ protocol sketched in Figure~\ref{fig:lin}
%solves the CPC Problem with tag
%complexity $f+2$, where $f$ is the maximal number of processes that
%may fail. It
allows controllers to use up to $f+2$ tags dynamically and in a coordinated fashion.
As we will also show in this section, there does not exist any
solution with less than $f+2$ tags.
Note that in the fault-free scenario ($f=0$), only one bit can be used
for storing the policy tag.

\vspace{1mm}
\noindent
\textbf{State machine.}
The protocol is built atop a replicated state machine
(implemented, \eg, using the
construction of~\cite{Her91}) that imposes a global
order on the policy updates and ensures a coordinated use and reuse of
the protocol tags.
%Without loss of generality,
For simplicity, we assume that the policies in the updates are uniquely identified.
%[[PK redundant
%Processes use the consensus abstraction to implement, in a fault-tolerant manner,
%the replicated state machine based on its sequential
%specification (see~\cite{Her91} for a construction).
%]]

The state machine we are going to use in our algorithm, and which we
call {\PS} (for \emph{Policy Serialization})
exports, to each controller $p_i$, two operations:

%\begin{itemize}
%\item 
\vspace{1.5mm}
\noindent $\bullet$
$\textit{push}(i,\pi)$, where $\pi$ is a policy, that always returns
  $\texttt{ok}$; % (a constant); %\ssnote{do we need to return
                 % something at all?}
%\item 

\vspace{0.5mm}
\noindent $\bullet$ $\textit{pull}(i)$ that returns $\bot$  or a tuple
  $(\pi,\textit{tag})$,
  where $\pi$ is a policy and
  $\textit{tag}\in\{0,\ldots, f+1\}$.
 %\ssnote{check values},
 % and $\textit{resolved}\in\{0,1\}$ is a binary variable.
%\end{itemize}
%
\vspace{1.5mm}

Intuitively, $p_i$ invokes $\textit{push}(i,\pi)$ to put policy $\pi$
in the queue of policies waiting to be installed;
and $p_i$ invokes $\textit{pull}(i)$ to fetch the next policy
to be installed. The invocation of \textit{pull} returns $\bot$ if all policies pushed so
far are already installed and there is an ``available'' tag (to be
explained below), otherwise it returns a tuple
$(\pi,\textit{tag})$, informing $p_i$ that policy $\pi$ should be
equipped with $\textit{tag}$.
%If $\textit{resolved}$ is true
% then at least one controller completed installing the
%  policy already and, as we shall see, the tag of the previously
%  installed policy can be reused.

%\end{itemize}
%
%Now we define the sequential behavior of {\PS}.
Let $S$ be a sequential execution of {\PS}.
Let $\pi_1,\pi_2,\ldots$ be the sequence of policies proposed in $S$
as arguments of the $\textit{push}()$ operations (in the order of appearance).
Let $(\pi_{i,1},\tau_{i,1}),(\pi_{i,3},\tau_{i,2}),\ldots$
be the sequence of non-$\bot$ responses to $\textit{pull}(i)$
operations in $S$ (performed by $p_i$).
If $S$ contains exactly $k$ \emph{non-trivial} (returning non-$\bot$ values)
$\textit{pull}(i)$ operations, then we say that $p_i$ \emph{performed $k$
non-trivial pulls in $S$}.
If $S$ contains $\textit{pull}(i)$ that returns $(\pi,t)\neq\bot$,
%where $t\neq \bot$ and $b\in\{0,1\}$,
followed by a subsequent $\textit{pull}(i)$, then we
say that $\pi$ is \emph{installed} in $S$. %\ssnote{I would prefer `installed' to 'resolved'. Petr?}

We say that $\tau_k$ is \emph{blocked} at the end of a finite history $S$
if $S$ contains $\textit{pull}(i)$ that
returns $(\pi_{k+1},\tau_{k+1}, 0)$ but does not contain a subsequent
$\textit{pull}(i)$.
In this case, we also say that $p_i$ \emph{blocks} tag $\tau_{k}$ at the
end of $S$.
Note that a controller installing policy $\pi_{k+1}$ blocks the tag associated
with the \emph{previous} policy $\pi_{k}$ (or the initially installed policy
in case $k=0$).
Now we are ready to define the sequential specification of {\PS}
via the following requirements on~$S$:
%(recall that
%$\pi_1,\pi_2,\ldots$ is the sequence of policies proposed in $S$):
%\begin{description}
%\item[Non-triviality:]

\vspace{1.5mm}
\noindent $\bullet$ \textbf{Non-triviality:}
If $p_i$ performed $k$ non-trivial pulls, then
a subsequent $\textit{pull}(i)$ returns $\bot$ if and only if
the pull operation  is preceded by at most $k$ pushes or  $f+1$ or more
policies are blocked in $S$.
%Also, $\textit{pull}(i)$ returns $(\pi_k,t,0)$ if and only if $\pi_k$
%is not yet installed in the prefix of $S$ ending with the pull operation.
In other words, the $k$-th pull of $p_i$ must return some
policy if at least $k$ policies were previously pushed and at most $f$
of them are blocked.
%Also, the flag \emph{resolved} associated with $\pi$ can be true if
%and only if at least one controller already asked for the subsequent policy
%in the queue.
%As we shall see, in our algorithm, this would mean that the controller
%has previously installed~$\pi$.

\vspace{1.5mm}
%\item[Agreement:]
\noindent $\bullet$ \textbf{Agreement:}
For all $k>0$, there exists $\tau_k\in\{0,\ldots,f+1\}$ such that if
controllers $p_i$ and $p_j$ performed $k$
nontrivial pulls, then $\pi_{i,k}=\pi_{j,k}=\pi_k$ and
$\tau_{i,k}=\tau_{j,k}=\tau_k$ for some $\tau_k$.
Therefore, the controllers compute the same order in which the
proposed policies must be installed, with the same sequence of tags.

\vspace{1mm}
%\item[Tag validity:]
\noindent $\bullet$ \textbf{Tag validity:}
For all $k$, $\tau_k$ is the minimal value in $\{0,\ldots,f+1\}-\{\tau_{k-1}\}$ that
is not blocked in $\{0,\ldots,n-1\}$ when the first $\textit{pull}(i)$
operation that returns $(\pi_k,\tau_k,r_k)$ is performed.
The intuition here is that the tags are chosen deterministically based
on all the tags that are not currently blocked.
Since, by the Non-triviality property, at most $f$ policies are blocked
in this case, $\{0,\ldots,f+1\}-\{\tau_{k-1}\}$ is non-empty.

\begin{wrapfigure}{r}{0.5\textwidth}
%\begin{figure}[tbp]
\hrule \vspace{-2mm}
 {\small
\begin{tabbing}
 bbb\=bb\=bb\=bb\=bb\=bb\=bb\=bb \=  \kill
Initially:\\
\> $\seq := \bot$;
     $\cur_i := \bot$\\
\\
\textbf{upon} \textit{apply}$(\tilde \pi)$\\
\nnll\> $cur_i := \tilde \pi$\\
\nnll\label{line:push}\> $\PS.push(i,\tilde \pi)$\\
\\
\textbf{do forever}\\
\nnll\label{line:pull}\> \textbf{wait until} $\PS.pull(i)$ returns $(\pi,t)\neq\bot$\\
\nnll\label{line:checkconflict}\> \textbf{if} ($\seq$ and $\pi$ conflict) \textbf{then}\\
\nnll\>\> $res := nack$\\
\nnll\> \textbf{else} \\
\nnll\>\> $\seq := \seq.(\pi,t)$\\
\nnll\label{line:wait3}\>\> \textbf{wait until} $\Tag(|\seq|-1)$ is not used \\
\nnll\label{line:install}\>\> \textit{install}$(\seq)$\\
\nnll\>\> $res := ack$\\
\nnll\label{line:result}\>\textbf{if} $\pi=\cur_i$ \textbf{then}
return $res$; $\cur_i:=\bot$
\end{tabbing}
\vspace{-2mm}
\hrule
}
\caption{The $\DPO$ algorithm: pseudocode for controller $p_i$.}
\label{fig:lin}
\end{wrapfigure}

%\end{description}
\vspace{2mm}
\noindent
In the following, we assume that a \emph{linearizable} $f$-resilient
implementation of  {\PS} is available~\cite{HW90}: any concurrent history of
the implementation is, in a precise sense, equivalent to a
sequential history that respects the temporal relations on operations
and every operation invoked by a correct controller returns, assuming
that at most $f$ controllers fail.
Note that the {\PS} machine establishes a total order on %composable
policies $(\pi_1,\textit{tag}_1), (\pi_2,\textit{tag}_2),\ldots$,
which we call the \emph{composition order} (the policy requests that
do not compose with a prefix of this order are ignored).

\vspace{1mm}
\noindent
\textbf{Algorithm operation.}
The algorithm is depicted in Figure~\ref{fig:lin} and operates as follows.
To install policy $\tilde \pi$, controller $p_i$ first
pushes $\tilde \pi$ to the policy queue by invoking ${\PS}.\textit{push}(i,\tilde\pi)$.

In parallel, to install its policy and help the others, the controller runs the
following task (Lines~\ref{line:pull}-\ref{line:result}).
First it keeps invoking  ${\PS}.\textit{pull}(i)$ until a
(non-$\bot$) value $(\pi_k,\tau_k)$ is returned (Line~\ref{line:pull}); here $k$ is the number
of nontrivial pulls performed by $p_i$ so far.
The controller checks if $\pi_k$
is not conflicting with previously installed policies (Line~\ref{line:checkconflict}), stored in
sequence $\textit{seq}$.
%Also if $r_k$ is $\textit{true}$ (\ie, $\pi_k$ has been already installed by another),
%$p_i$ simply waits until the result of $\pi_k$'s installation is
%received (Line~\ref{line:wait2}).
Otherwise, in Line~\ref{line:wait3}, $p_i$ waits until the traffic in the network only
carries tag $\tau_{k-1}$ (the tag $\tau_{k-2}$ used by the penultimate policy in $\seq$,
denoted $\Tag(|\seq|-1)$).
Here $p_i$ uses the \emph{oracle} (described in
Section~\ref{sec:model}) that produces the set of
currently active policies.
%Here we assume that initially all the traffic
%in the network is tagged with $\tau_0=0$ used by the initial policy
%$\pi_0$.

Then the controller tries to install $\pi_k$
  on all internal ports followed by the ingress ports, one by one, in
  a pre-defined order, employing the
``two-phase update'' strategy of~\cite{network-update} (Line~\ref{line:install}).
%Here the ingress ports should be updated to apply
%  a new tag-change rule to all injected packets in $\dom(\pi_k)$ (Line~\ref{line:install}).
The update of an internal port $p$ is performed using an atomic operation
that adds the rule associated with $\pi_k$ equipped with $\tau_k$ to the set of
rules currently installed on $p$.
The update on an ingress port $p$ simply replaces the currently
installed rule with a new rule tagging the traffic with $\tau_k$ which succeeds \emph{if and
only if} the port currently carries the policy tag $\tau_{k-1}$
(otherwise, the port is left untouched).
Once all ingress ports are updated, old rules are removed, one by
one, from the internal ports.
%Finally, the result of the installation is broadcast to all other
%controllers (Line~\ref{line:bcast}).
If $\pi_k$ happens to the policy currently proposed by $p_i$, the result is returned to the application.

Intuitively, a controller blocking a tag $\tau_k$ may still be involved in
installing $\tau_{k+1}$ and thus we cannot reuse $\tau_k$ for a policy other
than $\pi_k$. Otherwise, the slow controller may wake up and update a
port with an outdated rule.
But since a slow or faulty controller can block at most one tag, there
eventually must be at least one available tag in $\{0,\ldots,f+1\}-\{\tau_{k-1}\}$
when the first controller performs its $k$-th nontrivial pull.
In summary, we have the following result.
\begin{theorem}\label{thm:lin}
$\DPO$ solves the CPC Problem $f$-resiliently with tag complexity
  $f+2$ using $f$-resilient consensus objects.
\end{theorem}
A natural optimization of the $\DPO$ algorithm is to allow a
controller to broadcast the outcome of each complete policy
update. This way ``left behind'' controllers can catch up with the more
advanced ones, so that they do not need to re-install already installed policies.

Note that since in the algorithm, the controllers maintain a
total order on the set of policy updates that respects the order, we can easily extend it to encompass \emph{removals} of previously
installed policies.
To implement removals, it seems reasonable to assume that a removal request for a policy
$\pi$ is issued by the controller that has previously installed $\pi$.

The tag complexity of $\DPO$ is, in a strict sense, optimal.
Indeed, we now show that there exists no $f$-resilient CPC algorithm that uses
$f+1$ or less tags in any network.
By contradiction, for any such algorithm we construct a
network consisting of two ingress ports connected to $f$ consecutive
loops.
We then consider $f+2$ composable policies, $\pi_1,\ldots,\pi_{f+2}$,
that have overlapping domains but prescribe distinct paths.
%[[PK incorrect
%
Given that only $f+1$ tags are available, we can construct an
execution of the assumed algorithm in which one of $f$ ``slow but believed to faulty controllers'' wakes up 
and invalidates a more recently installed policy that uses the same
tag, contradicting the Consistency property of CPC.
%Given that only $f+1$ tags are available, we can construct an
%execution in which a policy update installing
%$\pi_i$ invalidates the previously installed $\pi_j$ by using the same
%tag, contradicting the Consistency property of CPC.
%]]
Thus:
%[[ No space to give a convincing sketch]]

\begin{theorem}\label{th:tagbound}
For each $f\geq 1$, there exists a network such that
any $f$-resilient CPC algorithm using %an arbitrary number of
$f$-resilient consensus objects
has tag complexity at least $f+2$.
\end{theorem}

%============================================================
\section{Related Work}\label{sec:related}
%============================================================

%
%[[PK not sure we need to be pedagogical here
%The implementation of consistent and concurrent read/write operations is one of the most fundamental tasks in distributed computing.
%Accordingly, our work builds upon classic problems such as consensus
%but also software transactional networks. Hence, in the following,
%we will first put our work into perspective with respect to the distributed computing literature.
%Subsequently, we will review related work in the field of Software Defined Networking.
%]]
%

% While the need for distributed SDN control planes is widely acknowledged in the
% networking community (\eg,~\cite{onix}),

%In the following, we review related work in the area of distributed computing,
%and then discuss related Software-Defined Networking literature.

\noindent\textbf{Distributed Computing.}
There is a long tradition of defining correctness of a concurrent system via
an equivalence to a sequential one~\cite{Pap79-serial,Lam79,HW90}.  The notion
of sequentially composable histories is reminiscent of
linearizability~\cite{HW90}, where a history of concurrently applied operations
is equivalent to a history in which the
operations are in a sequential order, respecting their real-time precedence.
Our sequentially composable histories impose requirements not
only on high-level invocations and responses, but also on the way the traffic
is processed. We require that the committed policies constitute a
conflict-free sequential history, but, additionally,  we expect that each
packet trace is consistent with a prefix of this history,
%[[PK more precise
%consisting of
containing all requests that were committed before the packet was injected.
%]]

The transactional interface exported by the CPC abstraction is inspired by the
work on speculative concurrency control using software transactional memory
(STM)~\cite{stm-st95}.
%, thus the  term \emph{software transactional networking}.
Our interface is however intended to model realistic network
management operations, which makes it simpler than more recent models of
dynamic STMs~\cite{dstm}.

\noindent\textbf{Software Defined Networking.}
%[[PK to save space
%There are already several (typically more systems oriented) papers
%motivating the need for a distributed control plane.
%]]
At the heart of Software-defined networking (SDN)
lies the decoupling of the system that makes decisions about where traffic is sent (the \emph{control plane})
from the underlying systems that forward traffic to the selected destination (the \emph{data plane}).
For an introduction to SDN as well as for a discussion of the differences to concepts such as active networks
(where packets carry code and which are hard to formally verify)
and protocols such as MPLS (which do not come with a software control plane and which do not allow
users to specify even basic consistency properties),
we refer the reader to~\cite{fabric}.
We believe that our distributed controller can be used together with other link virtualization technologies
which support tagging. 

Onix~\cite{onix} is among the earliest distributed SDN controller platforms.
Onix applies existing distributed systems techniques to build a Network
Information Base (NIB), \ie, a data structure that maintains a copy of the
network state, and abstracts the task of network state distribution from
control logic. However, Onix expects developers to provide the logic that is
necessary to detect and resolve conflicts of network state due to concurrent
control. In contrast, we study concurrent policy composition mechanisms that
can be leveraged by any application in a general fashion.

%It has been argued that distributing
%control is useful to reduce the latency of reactive control~\cite{ctrl-place},
%as certain SDN controllers (and hence functionality) can be placed closer to
%the switches~\cite{kandoo}.
%[[PK do not see how this is relevant
%Kotronis \etal~proposed a novel inter-domain routing
% architecture where the ``BGP''-route selection is out-sourced~\cite{outsourcing} to an external entity.
%]]
\ignore{
\mcnote{Removing to save space. We add it back for camera ready}
The costs of implementing logically centralized network control, providing
various levels of consistency, over a distributed system are subjected to a
sensitivity study in~\cite{log-cent}.
}
%[[PK to save space and stay in focus
%In~\cite{log-cent}, Levin~\etal~studied  the implications of the logically centralized abstraction
%provided by SDN and the resulting design choices inherent to the strong or
%eventual network view consistency. The authors demonstrate certain consequences of
%this design choice on the performance of a distributed load-balancing control
%application. Regarding wide-area SDN networks under multiple
%administrative controls, researchers have also pointed out
%interesting connections to the field of locality-sensitive computing.~\cite{hotsdn13loc}
%]]

For the case of a single controller, Reitblatt \etal~\cite{network-update}
formalized the notion of per-packet consistency and introduced the problem of
\emph{consistent network update}.
%, and described the \emph{two-phase update} technique, also used in our algorithms.
Mahajan and Wattenhofer~\cite{roger-hotnets} introduced several new variants
of network update problems, and presented more efficient, dependency-based
protocols. We complement this line of research by assuming a distributed
computing perspective, and by investigating robust and concurrent policy
installations. Our work also introduces the notion of tag complexity.
\ignore{
\mcnote{Removing to save space. We add it back for camera ready}
(The solutions in~\cite{roger-hotnets} require, in the asynchronous network,
an unbounded number of policy tags.)

Detection and prevention of SDN control plane (software) failures is considered
one of the main challenges facing SDNs. An early approach to systematic fault
reproduction was presented in~\cite{ofrewind}. Later work~\cite{nice} leverages
symbolic execution to systematically tackle the fault-tolerance aspect of the
software defined network control plane. Regarding automated policy composition,
our work builds upon the results Foster \etal~\cite{frenetic,pyretic}.
Indeed, the parallel policy composition of Pyretic (operator ``$|$'') is a
building block in the prototype implementation of our control plane.
}

\noindent
\textbf{Bibliographic Note.} In our SIGCOMM HotSDN workshop
paper~\cite{hotsdn13ccc}, we introduced the notion of software transactional
networking, and sketched a tag-based algorithm to consistently compose
concurrent network updates. However, the algorithm proposed there
is not robust to any controller failure, and features an exponential
tag complexity. (A simple corollary of the present paper is that the non-failure
setting can be solved with two tags only.)
%[[PK not sure it is true
%and is
%not robust to any controller failure.
%[[PK does not seem relevant
%(A protocol solving this problem
%with tag complexity two is given in Appendix~\ref{ssec:bit}.)
%]]

%provided by a network update middleware, and there is no
%formal discussion of what can and cannot be achieved under
%distributed network control.
%Moreover,~\cite{hotsdn13ccc} does not consider failures or asynchronous executions.
% The algorithm in~\cite{hotsdn13ccc} features an exponential tag complexity and
% is not robust to any controller failure.

%============================================================
\section{Concluding Remarks}\label{sec:conc}
%============================================================

We believe that our paper opens a rich area for future research,
and we understand our work as a first step towards a better understanding of
how to design and operate a robust SDN control plane.
% under different failure
%models.
%
As a side result, our model allows us to gain insights into minimal requirements
on the network that enable consistent policy updates: \eg, in Appendix~\ref{sec:weaker-port}, we
prove that consistent network updates are impossible if SDN ports do not support
atomic read-modify-write operations.

Our $\CPO$ and $\DPO$ algorithms highlight the fundamental trade-offs
between the concurrency of installation of policy updates and the
overhead on messages and switch memories.
Indeed, while being optimal in terms of tag complexity, $\DPO$
essentially reduces to installing updates sequentially.
Our initial concerns were resilience to failures and overhead, so our
definition of the CPC problem did not require any form of
``concurrent entry''~\cite{GME-00}.
But it is important to understand to which extent the concurrency of a
CPC algorithm can be improved, and we leave it to future research.

Another direction for future research regards more complex, non-commutative policy compositions:
while our protocol can also be used for, \eg, policy removals,
it will be interesting to understand how general such approaches are.
We have also started to develop a proof-of-concept prototype implementation of our distributed
control plane~\cite{ons13stn}.

%[[PK comment for now
%\section*{Acknowledgments}
%
%We would like to thank Yehuda Afek and Roy Friedman for interesting
%discussions.
%]]

%\newpage

\bibliographystyle{abbrv}
\bibliography{references}  % main.bib is the name of the Bibliography in this case

%\newpage

\begin{appendix}

\section{Impossibility for Weaker Port Model}\label{sec:weaker-port}

It turns out that SDN ports must support atomic policy updates (\ie, an atomic read-modify-write);
otherwise it is impossible to update a network consistently in the presence of even one crash failure.
Concretely, we assume here that a port can be accessed with two atomic
operations: $\textit{read}$ that returns the set
  of rules currently installed at the port and $\textit{write}$ that
  updates the state of the port with a new set of rules.
%[[PK updated for simplicity
% $\textit{read}$,
%$\textit{add-rule}$, and $\textit{remove-rule}$:
%
%A $\textit{read}(i)$ operation returns the set of rules representing the
%switch function of port $i$.
%An $\textit{add-rule}(i,r)$ operation adds a new rule $r$ to port $i$.
%A $\textit{remove-rule}(i,r)$ operation removes rule $r$
%from port $i$ (if present). \ssnote{We should say later when we remove rules in our protocols.}
%]]

\begin{theorem}\label{thm:impossible}
There is no solution to CPC using consensus objects that tolerates one or more crash failures.
\end{theorem}
\begin{proof}
By contradiction and assume that there is a $1$-resilient CPC
algorithm $A$ using consensus objects.

Consider a network including two ingress ports, $1$ and $2$, initially
configured to forward all the traffic to internal ports (we denote this policy by $\pi_0$).
Let processes $p_1$ and $p_2$ accept two policy-update requests $\textit{req}_1=\textit{apply}_1(\pi_1)$
and $\textit{req}_2=\textit{apply}_2(\pi_2)$, respectively, such that
$\pi_1$ is refined by $\pi_2$, \ie, $\pr(\pi_2)>\pr(\pi_1)$ and $\dom(\pi_2)\subset\dom(\pi_1)$,
and paths stipulated by the two policies to ingress ports $1$ and $2$
satisfy $\pi_1^{(1)}\neq \pi_2^{(1)}$ and $\pi_1^{(2)}\neq \pi_2^{(2)}$.

Now consider an execution of our $1$-resilient algorithm in which
$p_1$ is installing $\pi_1$ and $p_2$ takes no steps. Since the
algorithm is $1$-resilient, $p_1$ must eventually complete the
update even if $p_2$ is just slow and not actually faulty. Let us stop
$p_1$ after it has configured one of the ingress ports, say $1$, to
use policy $\pi_1$, and just before it changes the state of $2$ to use
policy $\pi_1$. Note that since $p_1$ did not witness a single step of
$p_2$ the configuration it is about to write to port $2$ only contains
the composition of $\pi_0$ and $\pi_1$.

Now let a given packet in $\dom(\pi_1)$ arrive at port $1$ and be
processed according to $\pi_1$.
We extend the execution with $p_2$ installing $\pi_2$ until both ports $1$
and $2$ are configured to use the composition
$\pi_0\cdot\pi_1\cdot\pi_2$. Such an execution exists, since the
algorithm is $1$-resilient and $\pi_1$ has been already applied to one
packet.
Therefore, by sequential composability,  the sequential equivalent of
the execution, both $\textit{apply}(\pi_1)$ and
$\textit{apply}(\pi_1)$
must appear as committed.

But now we can schedule the enabled step of $p_1$ to overwrite the
state of port $2$ with the ``outdated'' configuration that does not
contain $\pi_2$. From now on, every packet in $\dom(\pi_2)$ injected
at port $2$ is going to be processed according to $\pi_1$---a
contradiction to sequential composability.
\end{proof}

\section{Proofs}

\subsection{Proof of Theorem~\ref{thm:tags}}

The correctness of the algorithm is based on three simple arguments.

\begin{enumerate}
\item \emph{Global policy order:} The strict port order $\prec$
  guarantees that the equivalent sequential
history respects the total order of policy updates imposed by the ingress port of
lowest order. 

\item \emph{All-or-nothing semantics:} A policy which started taking effect at some ingress ports will eventually 
be installed at all ingress ports. This follows from the reliable broadcast implementation: the rebroadcasts ensure
that eventually, all processes will learn about (and help finish) the planned policy installation, even if the initiator failed before it
notified the other processes. 
% Stefan: do we need to remind the reader about commutativity and idempotency assumptions?

\item \emph{Consistency:} The proof of per-packet consistency is simple: a packet will be marked with an immutable
tag at its ingress port, and the tag defines a unique path in the
network, consistent with the corresponding policy. 
%\hfill $\Box$
\end{enumerate}

\subsection{Proof of Theorem~\ref{thm:lin}}

%We study the termination and consistency properties in turn.

\emph{Termination:} Consider any $f$-resilient execution $E$ of $\DPO$ and let $\pi_1,\pi_2,\ldots$ be the
sequence of policy updates as they appear in the
linearization of the state-machine operations in $E$.
Suppose, by contradiction, that a given process $p_i$ never completes
its policy update  $\pi$.
Since our state-machine $\PS$ is $f$-resilient, $p_i$ eventually completes its
$\textit{push}(i,\pi)$ operation.
Assume $\pi$ has order $k$ in the total order on push operations.
Thus, $p_i$ is blocked in processing some policy $\pi_{\ell}$, $1\leq \ell \leq
k$,  waiting in Lines \ref{line:pull} or ~\ref{line:wait3}.

Note that, by the Non-Triviality and Agreement properties of $\PS$,
when a correct process completes installing $\pi_{\ell}$,
eventually every other correct process completes installing $\pi_{\ell}$.
%This is because the correct processes perform policy updates in the
%same order. % and, once finished, broadcast the result to every other
%process (Line~\ref{line:bcast}).
Thus, all correct processes are blocked while processing $\pi$.
Since there are at most $f$ faulty processes, at most $f$ policies can
be blocked forever.
Moreover, since every blocked process has previously pushed a policy
update, the number of processes that try to pull proposed policy
updates
cannot exceed the number of previously pushed policies.
Therefore, by the Non-Triviality property of $\PS$,
eventually, no correct process  can be blocked forever in
Line~\ref{line:pull}.

%No correct process can be blocked
%in Line~\ref{line:wait2}. Indeed, by the Non-Triviality property,
%if a correct process pulled $(\pi_{\ell},\tau_{\ell},r)$ with $r=1$,  then $\pi_{\ell}$ has been previously installed by some
%process that must have broadcasted $(\pi_{\ell},\ell-1)$ (by the
%Agreement property of $\PS$, every correct process blocked
%while processing $\pi_{\ell}$ has $|S|-\ell-1$).

Finally, every correct process has previously completed installing
policy with tag $\tau_{\ell-1}$.
By the algorithm, every injected packet is tagged with $\tau_{\ell-1}$ and,
 eventually, no packet with a tag other
than $\tau_{\ell-1}$ stays in the network.
Thus, no  correct process can be blocked
in Line~\ref{line:wait3}---a contradiction, \ie, the algorithm
satisfies the Termination property of CPC.

\emph{Consistency:} To prove the Consistency property of CPC,  let $S$ be a sequential
history that respects the total order of policy updates determined by
the $\PS$. According to our algorithm, the response of each update in
$S$ is $\ack$ if and only if it does not conflict with the set of
previously committed updates in $S$.
Now since each policy update in $S$ is installed by the two-phase
update procedure using atomic read-modify-write update operations,
every packet injected to the network, after a policy update
completes, is processed according to the composition of
the update with all preceding updates.
Moreover, an incomplete policy update that manages to push the policy
into $\PS$ will eventually be completed by some correct process (due to the reliable broadcast implementation).
Finally, the per-packet consistency follows from the fact that packets will always respect
the global order, and are marked with an immutable
tag at the ingress port; the corresponding forwarding rules are never changed while packets are
in transit.

Thus, the algorithm satisfies the Consistency property of CPC.

\subsection{Proof of Theorem~\ref{th:tagbound}}
\label{app:proof:dpo}

Assume the network $T_f$ of two ingress ports $A$ and $B$, and $f+1$
``loops'' depicted in Figure~\ref{fig:tagbound} and consider a
scenario in which the controllers apply a sequence of policies defined
as follows.
%[[
%Suppose that the controllers concurrently apply
%composable policies as follows.
%]]
%Let $\pi_b$, where $b\in\{0,2^{f+1}-1\}$, denote a policy
%which specifies a path starting from the ingress port and following
%$b$ in the sense that bit $1$ in position $k$ of the binary
%representation of $b$ means that the path specified by $\pi_b$
%takes the lower path in loop $k$.
%For example the  path specified by $\pi_{1}$ goes through the
%upper paths in loops $f,\ldots,1$ and then takes the lower path
%in loop $0$.
Let $\pi_i$, $i=1,\ldots,f$, denote a policy
which, for each of the two ingress port,  specifies a path that in
every loop $\ell\neq i$ takes the upper path and in loop $i$ takes the
lower path (the dashed line in Figure~\ref{fig:tagbound}).
The policy $\pi_0$ specifies the paths that always go over the upper parts
of all the loops (the dashed line in Figure~\ref{fig:tagbound}).

%$b$ in the sense that bit $1$ in position $i$ of the binary
%representation of $b$ means that the path specified by $\pi_b$
%takes the lower path in loop $i$.
%For example the  path specified by $\pi_{1}$ goes through the
%upper paths in loops $f,\ldots,1$ and then takes the lower path
%in loop $0$.

We assume that for any two policies $\pi_i$ and $\pi_{j}$, such that
$0\geq j<j \geq f+1$, we have $\pr(\pi_{i})>\pr(\pi_{j})$
and $\dom(\pi_{i})\subset\dom(\pi_{j})$, \ie, all these the
policies  are composable, and  adding a new policy to the
composition makes the composed policy more refined.
Note that, assuming that only policies $\pi_i$, $i=0,\ldots,f+1$,
are in use, each of the ingress ports may only store one rule per tag that
forwards all the packets to the next branching port.
Intuitively, the only way to make sure that an injected packet is processed
according to a new policy in the set is to equip injected packets with
specific tags and forward them further.

Suppose that $0$ is the tag used for the initially installed $\pi_0$.
By induction on $i=1,\ldots,f+1$, we are going to show that any $f$-resilient CPC algorithm
on $T_f$ has a finite execution $E_i$
%(every policy update is complete it),
 at the end of which
(1) a composed policy $\pi_0\cdot\pi_1\cdots\pi_{i}$ is installed
and (2) there is a set of $i$ processes, $q_1,\ldots,q_{i}$, such that
each $q_{\ell}$, $\ell=1,\ldots,i$, is about to access
an ingress port with an update operation that, if the currently installed
rule uses $\ell-1$ to tag the injected packets, replaces it with a
rule that uses $\ell$ instead.

For the base case $i=1$, assume that $p_1$ proposes to install
$\pi_1$. Since the network initially carries traffic tagged $0$, the
tag used for the composed policy $\pi_0\cdot\pi_1$ must use a tag
different from $0$, without loss of generality we call it $1$.
There exists an execution in which
some process $q_1$ has updated the tag on one of the ingress port with
tag $1$ and is just about update the other port.
Now we ``freeze'' $q_1$ and let another process to complete the update
of the remaining ingress port.
Such an execution exists, since the protocol is $f$-resilient,
assuming that $f>0$ and, by the Consistency property of CPC,  any update that
affected the traffic must be eventually completed.
In the resulting execution $E_1$, $q_1$ is about to update
an ingress port to use tag $1$ instead of $0$ and the network operates
according to policy $\pi_0\cdot\pi_1$.

\begin{figure}[tbph]
  \centering
  \includegraphics[scale=0.35]{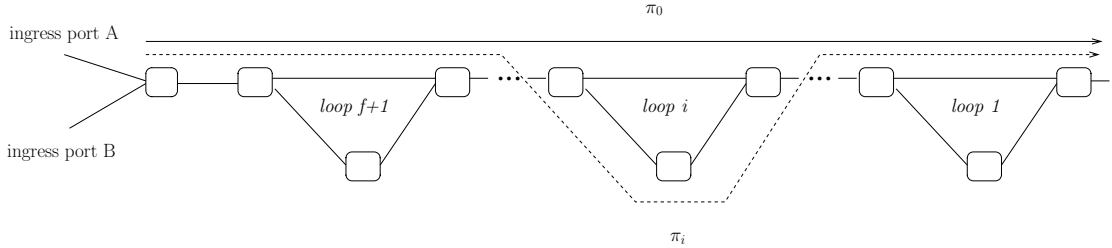}
  \caption{The $(f+1)$-loop network topology $T_f$.}
  \label{fig:tagbound}
\end{figure}

Now take $1<i\leq f+1$ and,
inductively, consider the execution $E_{i-1}$.
%in which processes $p_1,\ldots,p_{i-1}$ are just about to change
%the ingress port to use policy compositions based on
%$\pi_0,\pi_1,\pi_2,\pi_4,\ldots,p_{2^{i-1}}$ using tags
%$1,\ldots,i-1$ assuming tags $0,\ldots,.
Now suppose that some process in $\Pi-\{q_1,\ldots,q_{i-1}\}$  proposes to install
$\pi_i$.
Similarly, since the algorithm is $f$-resilient (and, thus,
$(i-1)$-resilient), there is an extension of $E_{i-1}$ in which no
process in $\{q_1,\ldots,q_{i-1}\}$ takes a step after $E_{i-1}$ and
eventually some process $q_i\notin \{q_1,\ldots,q_{i-1}\}$ updates one of the ingress ports to apply
$\pi_0\cdots\pi_i$ so that instead of the currently used tag $i-1$ a
new tag $\tau$ is used.
(By the Consistency property of CPC, $\pi_i$ should be composed with all policies $\pi_0,\ldots,\pi_{i-1}$.)

Naturally, the new tag $\tau$ cannot be $i-1$.
Otherwise, while installing $\pi_0\cdots\pi_i$,
either $q_i$ updates port $i$ before port $i-1$ and
some packet tagged $i$ would have to
take lower paths in both loops $i$ and $i-1$ (which does not
correspond to any composition of installed policies),
or $q_i$ updates port $i-1$ before $i$ and some packet would
have to take no lower paths at
all (which corresponds to
the policy $\pi_0$ later overwritten by $\pi_0\cdots\pi_{i-1}$).

Similarly, $\tau\notin\{0,\ldots,i-2\}$. Otherwise, once the installation
of $\pi_0\cdots\pi_i$ by $q_i$ is completed, we can wake up process
$p_{t+1}$ that would replace the rule of tag $\tau$ with a rule using
tag $\tau+1$, on one of the ingress ports. Thus, every packet injected at the
port would be tagged $\tau+1$. But this would violate the
Consistency property of CPC, because $\pi_0\cdots\pi_{i}$ using tag
$\tau$ is the most recently installed policy.

Thus, $q_i$, when installing $\pi_0\cdots\pi_{i}$, must use a tag not in $\{0,\ldots,i-1\}$,
say $i$. Now we let $q_i$ freeze just before it is about to install
tag $i$ on the second ingress port it updates.
Similarly, since $\pi_0\cdots\pi_i$ affected the traffic already on
the second port, there is an extended execution in which another
process in  $\Pi-\{q_1,\ldots,q_{i}\}$ completes the update and we get
the desired execution $E_i$.

In $E_{f+1}$ exactly $f+2$ tags are concurrently in use, which
completes the proof.

\end{appendix}

\end{document}